\documentclass[11pt, draftcls, onecolumn]{IEEEtran}
\usepackage[mathscr]{eucal}
\usepackage[cmex10]{amsmath}
\usepackage{epsfig,epsf,psfrag}
\usepackage{amssymb,amsmath,amsthm,amsfonts,latexsym}
\usepackage{amsmath,graphicx,bm,xcolor,url}
\usepackage{fixltx2e}
\usepackage{array}
\usepackage{verbatim}
\usepackage{bm}
\usepackage{algorithmic, cite}
\usepackage{algorithm}
\usepackage{verbatim}
\usepackage{textcomp}
\usepackage{mathrsfs}
\usepackage{epstopdf}

\catcode`~=11 \def\UrlSpecials{\do\~{\kern -.15em\lower .7ex\hbox{~}\kern .04em}} \catcode`~=13 

\allowdisplaybreaks[4]

\newcommand{\calJ}{\mathcal{J}}

\newcommand{\calP}{\mathcal{P}}

\newcommand{\calX}{\mathcal{X}}

\newcommand{\bbN}{\mathbb{N}}

\DeclareMathAlphabet{\mathbsf}{OT1}{cmss}{bx}{n}
\DeclareMathAlphabet{\mathssf}{OT1}{cmss}{m}{sl}

\DeclareSymbolFont{bsfletters}{OT1}{cmss}{bx}{n}  
\DeclareSymbolFont{ssfletters}{OT1}{cmss}{m}{n}
\DeclareMathSymbol{\bsfGamma}{0}{bsfletters}{'000}
\DeclareMathSymbol{\ssfGamma}{0}{ssfletters}{'000}
\DeclareMathSymbol{\bsfDelta}{0}{bsfletters}{'001}
\DeclareMathSymbol{\ssfDelta}{0}{ssfletters}{'001}
\DeclareMathSymbol{\bsfTheta}{0}{bsfletters}{'002}
\DeclareMathSymbol{\ssfTheta}{0}{ssfletters}{'002}
\DeclareMathSymbol{\bsfLambda}{0}{bsfletters}{'003}
\DeclareMathSymbol{\ssfLambda}{0}{ssfletters}{'003}
\DeclareMathSymbol{\bsfXi}{0}{bsfletters}{'004}
\DeclareMathSymbol{\ssfXi}{0}{ssfletters}{'004}
\DeclareMathSymbol{\bsfPi}{0}{bsfletters}{'005}
\DeclareMathSymbol{\ssfPi}{0}{ssfletters}{'005}
\DeclareMathSymbol{\bsfSigma}{0}{bsfletters}{'006}
\DeclareMathSymbol{\ssfSigma}{0}{ssfletters}{'006}
\DeclareMathSymbol{\bsfUpsilon}{0}{bsfletters}{'007}
\DeclareMathSymbol{\ssfUpsilon}{0}{ssfletters}{'007}
\DeclareMathSymbol{\bsfPhi}{0}{bsfletters}{'010}
\DeclareMathSymbol{\ssfPhi}{0}{ssfletters}{'010}
\DeclareMathSymbol{\bsfPsi}{0}{bsfletters}{'011}
\DeclareMathSymbol{\ssfPsi}{0}{ssfletters}{'011}
\DeclareMathSymbol{\bsfOmega}{0}{bsfletters}{'012}
\DeclareMathSymbol{\ssfOmega}{0}{ssfletters}{'012}

\newcommand{\hatB}{\hat{B}}

\newcommand{\tilB}{\tilde{B}}

\newcommand{\hatG}{\hat{G}}

\DeclareMathOperator*{\argmax}{arg\,max}
\DeclareMathOperator*{\argmin}{arg\,min}

\DeclareMathOperator{\var}{\mathsf{Var}}

\newtheorem{theorem}{Theorem} 
\newtheorem{lemma}[theorem]{Lemma}

\newtheorem{proposition}[theorem]{Proposition}
\newtheorem{corollary}[theorem]{Corollary}
\newtheorem{definition}{Definition}

\newtheorem{remark}{Remark}

\newcommand{\qednew}{\nobreak \ifvmode \relax \else
      \ifdim\lastskip<1.5em \hskip-\lastskip
      \hskip1.5em plus0em minus0.5em \fi \nobreak
      \vrule height0.75em width0.5em depth0.25em\fi}

\usepackage[english]{babel}
\usepackage{amsmath,amsthm}
\usepackage{amsfonts}
\usepackage{graphicx}
\usepackage{verbatim}
\usepackage{amssymb}
\usepackage{epstopdf}
\usepackage{hyperref}
\usepackage{bbm}
\usepackage{cite,bbm,graphicx,amsmath,amssymb,mathrsfs,epsf} \usepackage{multirow}
\usepackage{subfigure,cite,graphicx,epsfig,amsmath,amssymb,mathrsfs,epsf,graphics,color,enumerate}

\def \E{\operatorname{E}}

\def \var{\operatorname{Var}}

\begin{document}
\title{Exact Moderate Deviation Asymptotics in Streaming Data Transmission}
\author{\IEEEauthorblockN{Si-Hyeon Lee, Vincent Y. F. Tan, and Ashish Khisti}\thanks{S.-H. Lee and A. Khisti are with the Department of Electrical and Computer Engineering,
University of Toronto, Toronto, Canada (e-mail: sihyeon.lee@utoronto.ca; akhisti@comm.utoronto.ca).  V.~Y.~F.~Tan is with the Department of Electrical and Computer Engineering
and the Department of Mathematics, National University of Singapore,
Singapore (e-mail: vtan@nus.edu.sg). The work of V. Y. F. Tan is supported in part by a Ministry of Education Tier 2 grant (grant number R-263-000-B61-113).}}

\maketitle

\begin{abstract}
In this paper, a streaming transmission setup is considered where an encoder observes a new message in the beginning of each block and a decoder sequentially decodes each message after a delay of $T$ blocks. In this streaming setup, the fundamental interplay between the coding rate, the error probability, and the blocklength in the moderate deviations regime is studied. For output symmetric channels, the moderate deviations constant is shown to improve over the block coding or non-streaming setup by exactly a factor of $T$ for a certain range of moderate deviations scalings. 
For the converse proof, a more powerful decoder to which some extra information is fedforward is assumed. The error probability is bounded first for an auxiliary channel and this result is translated back to the original channel by using a newly developed change-of-measure lemma, where the speed of decay of the remainder term in the exponent is carefully characterized.
For the achievability proof, a known coding technique that involves a joint encoding and decoding of fresh and past messages is applied with some manipulations in the error analysis. 
\end{abstract}
\begin{keywords}
Streaming transmission, moderate deviations, discrete memoryless channel, converse, change-of-measure
\end{keywords}

\IEEEpeerreviewmaketitle

\section{Introduction}
\begin{table}[t] \caption{Asymptotic behaviors in three regimes} \label{table:regimes}
\begin{center} 
   \begin{tabular}{|l || c | c|c|} 
 \hline  
Regime & Large deviations & Central limit &  Moderate deviations \\
\hline
Operating rate & $R<C$ & $R=C-\Theta(n^{-1/2})$  &  $R=C-\Theta(n^{-t})$ for $0<t<\frac{1}{2}$   \\
\hline
Error probability & Exponentially decaying  & Non-vanishing & Subexponentially decaying \\
\hline
\end{tabular} 
\end{center} 
\end{table}
In his pioneering work \cite{Shannon:48}, Shannon formulated the channel coding problem and characterized the maximum rate such that the probability of error can be driven to zero as the blocklength increases. Since Shannon's work, a vast body of literature has followed on the fundamental interplay between the coding rate, the error probability, and the blocklength, which can provide more refined insights for reliable communication systems. 
One approach to characterize the fundamental interplay  is to study the best exponential decay rate of the error probability (so-called \emph{error exponent}) for a given rate. Classical results characterized the best error exponents for a large class of channels  \cite{Gallager:65,ShannonGallager:67,Gallager:68,CsiszarKorner:11}.
Another approach is to fix the error probability at a non-vanishing quantity and study the best (largest) achievable rate for information transmission. Strassen \cite{Strassen} considered discrete memoryless channels (DMCs) and showed that the rate backoff from capacity scales as $\frac{1}{\sqrt{n}}$ with the constant of proportionality related to the so-called {\em dispersion}~\cite{PolyanskiyPoorVerdu:10}.  Polyanskiy {\em et al.} \cite{PolyanskiyPoorVerdu:10}   refined the asymptotic expansions and also compared the normal approximation to the finite blocklength (non-asymptotic) fundamental limits. 
For practical code design, it would be more relevant to simultaneously require the rate to approach to capacity and the error probability to decay to zero.  Altu\u{g} and Wagner~\cite{AltugWagner:14} established the best decay rate of the error probability when the rate approaches to the capacity strictly slower than $\frac{1}{\sqrt{n}}$. 
Polyanskiy and Verd\'u~\cite{PolyanskiyVerdu:10} relaxed some assumptions in the conference version of Altu\u{g} and Wagner's work~\cite{altug2010moderate}. 
In the aforementioned three approaches, the asymptotic behaviors of the coding rate and error probability in the blocklength are closely related to the large deviations, central limit, and moderate deviations theorems \cite{DemboZeitouni:09}, respectively, and hence the regime considered in each approach is often named after the related theorem. Table~\ref{table:regimes} summarizes the asymptotic behaviors in the three regimes.\footnote{In this paper, we use the usual asymptotic notations $o(\cdot), O(\cdot), \Theta(\cdot), \omega(\cdot)$, and $\Omega(\cdot)$ (see e.g., \cite{Knuth:76}) with the additional restriction that the sequences in them are positive.}

In addition to the block coding setup, it is also of practical interest to study a streaming transmission setup. In this setup, the sender must encode a stream of messages in a sequential fashion and the receiver must also decode the stream of messages in order. Some natural applications include control systems and multimedia applications.
Such a streaming setup is fundamentally different from the block coding setup as different messages have different decoding deadlines, yet overlapping transmission durations.  In the large deviations regime, the  streaming transmission has been studied in e.g., \cite{Schulman:96,sahai_thesis, SukhavasiHassibi:11,KhistiDraper:14,DraperChangSahai:14,DraperKhisti:11,Sahai:08}. The coding schemes are based on an approach known as tree coding and its variants. The only work that treats the converse is \cite{Sahai:08} for a bit-wise  setup. 
On the other hand, the streaming transmission in the moderate deviations and central limit regimes was first considered in \cite{LeeTanKhisti:arxiv15} for a streaming scenario where an encoder observes a new message in the beginning of each block and a decoder decodes each message after a delay of $T$ blocks. 
The work \cite{LeeTanKhisti:arxiv15} showed the following achievability results: (i) in the moderate deviations regime, the moderate deviations constant improves at least by a factor of $T$ and (ii) in the central limit regime, the dispersion is improved (reduced) by approximately a factor of $\sqrt{T}$ for a wide range of channel parameters. 
To the best of our knowledge, however, there has been no prior work on  converse parts for the streaming transmission in the moderate deviations and central limit regimes, and thus the characterization of the exact asymptotic behavior in these two regimes remains open.

In this paper, we characterize the exact moderate deviation asymptotics for streaming transmission over output symmetric channels for a certain range of moderate deviations scalings. Our streaming setup is the same as that in \cite{LeeTanKhisti:arxiv15} except the following  differences: (i) an additional parameter corresponding to the total number of streaming messages is introduced and (ii) the maximal probability of error over streaming messages is considered.\footnote{In \cite{LeeTanKhisti:arxiv15}, the number of streaming messages is assumed to be infinite and the average probability of error over streaming messages is considered.} 
Our results show that the moderate deviations constant for output symmetric channels improves exactly by a factor of $T$ compared to classical channel coding for a certain range of moderate deviations scalings under some mild conditions on the number of streaming messages. 
The achievability part of our result can be proved by manipulating the result in \cite{LeeTanKhisti:arxiv15} taking into account the aforementioned differences. 
Hence, our contribution is more on the converse part. We prove the converse for a more powerful decoder to which some extra information is fedforward. 
The converse proof consists of the following three steps: (i) prove that for such a feedforward decoder, it suffices to utilize the channel output sequences only in recent $T$ blocks, (ii)  lower bound the maximal error probability over a certain number of messages under an auxiliary channel, and (iii) translate the result back to the original channel by using a change-of-measure technique. This flow of the proof is similar with that in \cite[Section IV]{Sahai:08}. However, due to the inherent differences between our problem setting and that of \cite[Section IV]{Sahai:08}, our proof involves novel technical treatments. Most importantly, we are interested in the moderate deviations regime, while the work \cite[Section IV]{Sahai:08} assumes the large deviations regime. Thus, we need to delicately balance the scaling of the parameters involved in the aforementioned three steps so that those parameters have negligible impact on both the rate backoff and the moderate deviations constant. In particular, we establish a change-of-measure lemma in the moderate deviations regime where the speed of decay of the remainder term in the exponent (which affects the moderate deviation constant) is carefully characterized. In addition, since the work \cite[Section IV]{Sahai:08} analyzes the bit-wise error under the bit-wise encoding and decoding operations, we develop proof techniques adapted to the message-wise error under the block-wise operations. 

The rest of this paper is organized as follows. We formally state our streaming setup in  Section~\ref{sec:model} and present the main result in Section~\ref{sec:main}. The converse and achievability parts are proved in Sections~\ref{sec:converse} and~\ref{sec:achievability}, respectively. We conclude this paper in Section~\ref{sec:conclusion}.

\subsection{Notation}
For two integers $i$ and $j$, $[i:j]$ denotes the set $\{i,i+1,\cdots, j\}$. For constants $x_1,\cdots, x_k$ and $S\subseteq [1:k]$, $x_S$  denotes the  vector $(x_j: j\in S)$ and $x^j_i$ denotes $x_{[i:j]}$ where the subscript is omitted when $i=1$, i.e., $x^j=x_{[1:j]}$. 
This notation is naturally extended for vectors $\mathbf{x}_1,\cdots, \mathbf{x}_k$, random variables $X_1,\cdots, X_k$, and random vectors $\mathbf{X}_1, \cdots, \mathbf{X}_k$. All logs are to base 2.
For a DMC $(\mathcal{X},\mathcal{Y}, \{W(y|x): x\in \mathcal{X}, y\in \mathcal{Y}\} )$ and an input distribution $P\in \mathcal{P}(\mathcal{X})$, where $\mathcal{P}(\mathcal{X})$ denotes the set of all probability distributions on $\mathcal{X}$, we use the following standard notation and terminology in information theory: 
\begin{itemize}
\item Type of a vector $x^l$ of length $l$:
\begin{align}
P_{x^l}\in \mathcal{P}(\mathcal{X}) \mbox{ such that }P_{x^l}(x)=\frac{N_x(x^l)}{l} \mbox{ for } x\in \mathcal{X},
\end{align}
where $N_{x}(x^l)$ denotes the number of occurrences of $x$ in $x^l$. 

\item Information density:
\begin{align}
i(x;y):=\log \frac{W(y|x)}{PW(y)}, \label{eqn:info_dst}
\end{align}
where $PW(y):=\sum_{x\in \mathcal{X}}P(x)W(y|x)$ denotes the output distribution. We note that $i(x;y)$ depends on $P$ and $W$ but this dependence is suppressed. The definition \eqref{eqn:info_dst} can be generalized  for two vectors $x^l$ and $y^l$ of length $l$ as follows:  
\begin{align}
i(x^l;y^l):=\sum_{j=1}^l i(x_j;y_j).
\end{align}

\item Mutual information:
\begin{align}
I(P,W)&:=\E[i(X;Y)]\\
&=\sum_{x\in \mathcal{X}}\sum_{y\in \mathcal{Y}} P(x)W(y|x)\log\frac{W(y|x)}{PW(y)}.
\end{align}

\item Unconditional information variance:
\begin{align}
U(P,W)&:=\var[i(X;Y)].
\end{align}

\item Conditional information variance: 
\begin{align}
V(P,W)&:=\E[\var[i(X;Y)|X]].
\end{align} 

\item Capacity: 
\begin{align}
C=C(W):=\max_{P\in \mathcal{P}(\mathcal{X})} I(P,W).
\end{align}

\item Set of capacity-achieving input distributions: 
\begin{align}
\Pi=\Pi(W):=\{P\in \mathcal{P}(\mathcal{X}): I(P,W)=C(W)\}.
\end{align}

\item Channel dispersion:
\begin{align}
\nu=\nu(W)&:=\min_{P\in \Pi} V(P,W) \label{eqn:dispersion}\\
&\overset{(a)}{=}\min_{P\in \Pi} U(P,W),
\end{align}
where $(a)$ is from \cite[Lemma 62]{PolyanskiyPoorVerdu:10}, where it is shown that $V(P,W)=U(P,W)$ for all $P\in \Pi$. 

\item Haroutunian exponent at rate $R$:
\begin{align}
E^+(R)&:=\min_{V:C(V)\leq R} \max_{P\in \mathcal{P}(\mathcal{X})} D(V\|W|P)\\
&= \min_{V:C(V)\leq R} \max_{x\in \mathcal{X}} D(V(\cdot|x)\|W(\cdot|x)),
\end{align}
where $D(V(\cdot|x)\|W(\cdot|x))$ and  $D(V\|W|P)$ are the divergence and the conditional divergence, respectively, defined as 
\begin{align}
D(V(\cdot|x)\|W(\cdot|x))&:=\sum_{y\in \mathcal{Y}}V(y|x)\log \frac{V(y|x)}{W(y|x)} \\
D(V\|W|P)&:= \sum_{x\in \mathcal{X}}P(x)D(V(\cdot|x)\|W(\cdot|x)).
\end{align}
In \cite{CsiszarKorner:11, Haroutunian:77}, it is shown that $E^+(R)$ is an upper bound on the block-coding error exponent with fixed-length coding and noiseless output feedback. 

\item Sphere-packing exponent at rate $R$:
\begin{align}
E_{\mathrm{SP}}(R)&:= \max_{P\in \mathcal{P}(\mathcal{X})}\min_{V:I(P,V)\leq R} D(V\|W|P).
\end{align}
$E_{\mathrm{SP}}(R)$ is known to be an upper bound on the block-coding error exponent without feedback. It is clear that $E_{\mathrm{SP}}(R)\leq E^+(R)$. It is known  that $E_{\mathrm{SP}}(R)=E^+(R)$ for output symmetric DMCs, where a DMC is called \emph{output symmetric} according to \cite{Gallager:68} if $\mathcal{Y}$ can be partitioned into disjoint subsets in such a way that for each subset, the matrix of transition probabilities has the property that each row is a permutation of each other row and each column is a permutation of each other column. 
 
\end{itemize}

\section{Model} \label{sec:model}
Consider a DMC $(\mathcal{X},\mathcal{Y}, \{W(y|x): x\in \mathcal{X}, y\in \mathcal{Y}\} )$. 
For block channel coding, a code is usually defined with three parameters, i.e., the blocklength, the cardinality of message (or rate), and the probability of error. For a streaming setup, we introduce two more parameters  corresponding to the decoding delay and the number of total streaming messages. Formally, a streaming code  is defined as follows:
\begin{definition}[Streaming code]	\label{def:basic}
An $(n,M,\epsilon,T,S)$-streaming code consists of
\begin{itemize}
\item a sequence of messages $\{G_k\}_{k\in[1:S]}$ each distributed uniformly over $\mathcal{G}:= [1:M]$,
\item a sequence of encoding functions $\phi_k: \mathcal{G}^{\min\{k,S\}}\rightarrow \mathcal{X}^n$ for $k\in [1:S+T-1]$ that maps the message sequence $G^{\min\{k,S\}}\in  \mathcal{G}^{\min\{k,S\}}$ to the channel input codeword $\mathbf{X}_k \in \mathcal{X}^n$, and
\item a sequence of decoding functions $\psi_k: \mathcal{Y}^{(k+T-1)n}\rightarrow \mathcal{G}$ for $k\in [1:S]$ that maps the channel output sequences $\mathbf{Y}^{(k+T-1)} \in \mathcal{Y}^{(k+T-1)n}$ to a message estimate $\hat{G}_k\in \mathcal{G}$,
\end{itemize}
that satisfies  
\begin{align}
\max_{k\in [1:S]}\Pr(\hat{G}_k\neq G_k) \leq \epsilon, 
\end{align}
i.e., the maximal probability of error over all $S$ messages does not exceed $\epsilon$. 
\end{definition}

\begin{figure}[t]
 \centering
  {
  \includegraphics[width=140mm]{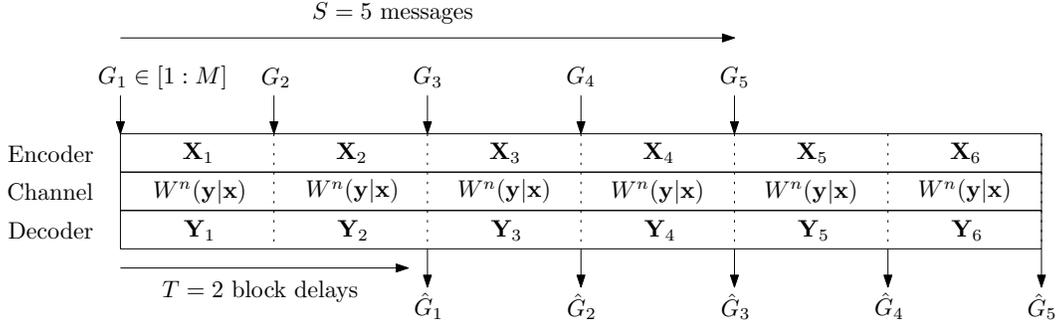}}
  \caption{Our streaming setup is illustrated  for the case of $T=2$ and $S=5$. A total of five messages ($S=5$) are sequentially encoded and are sequentially decoded after the delay of two blocks ($T=2$). } \label{fig:model}
\end{figure} 
For notational convenience, let $T_k$ denote $k+T-1$ for two positive integers $T$ and $k$. Fig. \ref{fig:model} illustrates our streaming setup for the case of $T=2$ and $S=5$. Since $S=5$, a total of five messages are sequentially encoded and decoded. In the beginning of block $k\in [1:5]$, the encoder receives a new message $G_k$ and generates a codeword $\mathbf{X}_k$ as a function of all the past and current messages $G^k$. In block 6, there is no new message and the encoder generates a codeword $\mathbf{X}_6$ as a function of all the past messages $G^5$. The encoder transmits $\mathbf{X}_k$ over the channel in block $k\in [1:6]$. Since $T=2$, the decoder decodes message $G_k$ for $k\in [1:5]$ at the end of block $k+1$, as a function of all the past received channel output sequences $\mathbf{Y}^{k+1}$. 

In this paper, we are interested in the following fundamental limit on the error probability:
\begin{align}
\epsilon^*(n,M,T,S)=\inf\{\epsilon: \exists (n,M,\epsilon,T,S)\mbox{-streaming code} \}.
\end{align}


\section{Main Result} \label{sec:main}
The following theorem presents the main result of this paper on the optimal behavior of $\epsilon^*(n,M,T,S)$ in the moderate deviations regime.  
\begin{theorem} \label{thm:match} For an output symmetric  DMC $(\mathcal{X},\mathcal{Y}, \{W(y|x): x\in \mathcal{X}, y\in \mathcal{Y}\} )$ with $\nu>0$, consider sequences $M_n$ and $S_n$ such that $\log M_n=nC-n^{1-t}$ and $S_n=\omega(n^t) \cap \exp\{o(n^{1-2t})\}$ for $0<t<\frac{1}{3}$. Then, 
\begin{align}
\lim_{n\rightarrow \infty}-\frac{1}{n^{1-2t}}\log \epsilon^*(n,M_n,T,S_n)=\frac{T}{2\nu}.\label{eqn:matching}
\end{align}
\end{theorem}
We note that the range of $S_n$ in Theorem \ref{thm:match} is quite extensive  since the order of $\exp\{n^{1-2t}\}$ is much larger than that of $n^t$. Theorem~\ref{thm:match} states that for an output symmetric DMC in a streaming setup with such a broad range of $S_n$, the moderate deviations constant\footnote{The moderate deviations constant is defined as the LHS of \eqref{eqn:matching} if the limit exists, see e.g., \cite[Definition 4]{ZhouTanMotani:arxiv16}. Our result shows that the limit exists for the range $(0,\frac{1}{3})$ of the moderate deviations scalings.} improves by a factor of $T$ for the range $(0,\frac{1}{3})$ of the moderate deviations scalings, which is a smaller set of scalings relative to the typical range  $(0,\frac{1}{2})$ (cf. Table \ref{table:regimes}). 
The converse and the achievability of Theorem~\ref{thm:match} are established by the following two propositions, respectively. The proofs are provided in Sections~\ref{sec:converse} and~\ref{sec:achievability}, respectively.  

\begin{proposition} \label{prop:converse}
For an output symmetric  DMC $(\mathcal{X},\mathcal{Y}, \{W(y|x): x\in \mathcal{X}, y\in \mathcal{Y}\} )$ with $\nu>0$, any sequence of $(n,M_n,\epsilon_n,T,S_n)$-streaming codes such that $\log M_n=nC-n^{1-t}$  and $S_n=\omega(n^t)$ for $0<t<\frac{1}{3}$ should satisfy   
\begin{align}
\limsup_{n\rightarrow \infty} -\frac{1}{n^{1-2t}} \log \epsilon_n \leq \frac{T}{2\nu}. \label{eqn:MD_conv_N}
\end{align}
\end{proposition}

\begin{proposition}\label{prop:achievability}
For a DMC $(\mathcal{X},\mathcal{Y}, \{W(y|x): x\in \mathcal{X}, y\in \mathcal{Y}\} )$ with $\nu>0$, there exists a sequence of $(n, M_n, \epsilon_n, T, S_n)$-streaming codes such that  $\log M_n=nC-n^{1-t}$, $S_n=\exp\{o(n^{1-2t})\}$, and $\epsilon_n$ satisfies
\begin{align}
\liminf_{n\rightarrow \infty} -\frac{1}{n^{1-2t}} \log \epsilon_n \geq \frac{T}{2\nu} \label{eqn:MD_achiev}
\end{align}
 for $0<t<\frac{1}{2}$.
\end{proposition}
\begin{remark}
 The condition $S_n=\omega(n^t)$ in Proposition \ref{prop:converse}  is related to the fact that the  backoff from capacity is $n^{-t}$.  An extreme case of $S_n=\omega(n^t)$ is the usual streaming setup \cite{Schulman:96,sahai_thesis,Sahai:08,SukhavasiHassibi:11,KhistiDraper:14,DraperKhisti:11,DraperChangSahai:14} in which the total number of streaming messages is infinite. On the other hand, the condition $S_n=\exp\{o(n^{1-2t})\}$ in Proposition~\ref{prop:achievability} is related to the fact that  the error probability decays as  $\exp\{-\Theta(n^{1-2t})\}$.
The scenario in which we decode a constant number of streaming messages is an extreme case of $S_n=\exp\{o(n^{1-2t})\}$. 
 
\end{remark}


\section{Converse} \label{sec:converse}
\begin{proof}[Proof of Proposition \ref{prop:converse}]
Consider an output symmetric  DMC $(\mathcal{X},\mathcal{Y}, \{W(y|x): x\in \mathcal{X}, y\in \mathcal{Y}\} )$ with $\nu>0$ and sequences $M_n$ and $S_n$ such that $\log M_n=nC-n^{1-t}$  and $S_n=\omega(n^t)$ for $0<t<\frac{1}{3}$. Let $R_n:=\frac{1}{n}\log M_n=C-n^{-t}$ and $\zeta:=\frac{1}{2}(\frac{1}{3}-t)>0$.  Since the proof is immediate from \cite{AltugWagner:14} for $T=1$, we assume that $T\geq 2$. 

Let us first present a sketch of the proof in the following that consists of three parts:
\begin{enumerate}[(i)]
\item We prove the converse for a more powerful decoder to which some extra information is fedforward. In Section \ref{subsec:fdf}, we present a formal definition of this feedforward decoder (Definition~\ref{def:ff_d}) and show that it is without loss of generality to assume a feedforward decoder that utilizes the channel output sequences only in recent $T$ blocks (Lemma~\ref{lemma:ff_recent}). Then, for such a feedforward decoder, the error probability of the $k$-th message is expressed in terms of some conditional probabilities of the channel output sequences in the $T$ blocks from the $k$-th block (Eq. \eqref{eqn:eq_err}).  

\item In Section \ref{subsec:lb}, we lower bound the maximal error probability over a certain number $S_n^*$ of messages under an auxiliary channel $V_n^*$ (Lemma~\ref{lemma:lb}).
We denote by $k^*$ the message index that contributes to the maximal error probability over $S_n^*$ messages under the auxiliary channel $V_n^*$.  
Then, the lower bound in Lemma~\ref{lemma:lb} is interpreted with respect to the conditional probabilities that are involved with (in the sense of Eq. \eqref{eqn:eq_err}) the error probability of the $k^*$-th message under the auxiliary channel $V_n^*$ (Corollary~\ref{corollary:proportion}).  

\item In Section \ref{subsec:cm}, based on the result in Corollary~\ref{corollary:proportion} under the auxiliary channel $V_n^*$, we  derive a lower bound on the error probability of the $k^*$-th message under the true channel $W$ by applying a technique of change-of-measure from  $V_n^*$ to $W$ (Lemma~\ref{lemma:ch_m}). It turns out that the sphere packing exponent is involved in the exponent of the resultant lower bound. By using an asymptotic bound on the sphere packing exponent (Lemma~\ref{lemma:sp}), an upper bound on the moderate deviations constant is derived. 
\end{enumerate}

Keeping this in mind, the detailed proof is provided in the following. 
\subsection{Feedforward decoder with an optimal sequence of decoding functions} \label{subsec:fdf}
We prove the converse for the following more powerful decoder that has knowledge of  additional information.   
\begin{definition}[Feedforward decoder] \label{def:ff_d} A feedforward decoder has a sequence of decoding functions $\psi_k^f: \mathcal{G}^{k-1}\times \mathcal{Y}^{T_kn} \rightarrow \mathcal{G}$ for $k\in [1:S_n]$ that maps the previous messages $G^{k-1}$ and the channel output sequences $\mathbf{Y}^{T_k} \in \mathcal{Y}^{T_kn}$ to a message estimate $\hat{G}_k\in \mathcal{G}$. 
\end{definition}

The following lemma states that it suffices for a feedforward decoder to consider decoding functions that utilize the channel output sequences only in recent $T$ blocks. The proof is relegated at the end of this section. 
\begin{lemma} \label{lemma:ff_recent} For a feedforward decoder, there exists a sequence of decoding functions $\psi_k^*: \mathcal{G}^{k-1}\times \mathcal{Y}^{Tn} \rightarrow \mathcal{G}$ for $k\in [1:S_n]$ that maps the previous messages $G^{k-1}$ and the recent $T$-block channel output sequences $\mathbf{Y}^{T_k}_{k} \in \mathcal{Y}^{Tn}$ to a message estimate $\hat{G}_k\in \mathcal{G}$ such  that 
\begin{align}
\Pr\Big(G_k\neq \psi_k^*(G^{k-1}, \mathbf{Y}^{T_k}_k)\Big)\leq \Pr\Big(G_k\neq \psi_k^f(G^{k-1}, \mathbf{Y}^{T_k})\Big)
\end{align} 
for any sequence of decoding functions $\psi_k^f$ for $k\in [1:S_n]$. 
\end{lemma}

Hence, without loss of generality, we prove \eqref{eqn:MD_conv_N} for a sequence of $(n,M_n,\epsilon_n,T,S_n)$-streaming codes with a feedforward decoder that has an optimal sequence of decoding functions  $\psi_k^*$ for $k\in [1:S_n]$ described in Lemma~\ref{lemma:ff_recent}. Let $A_{k}(g^{k})$ for $k\in [1:S_n]$ and $g^k\in \mathcal{G}^k$ denote the set of channel output sequences $\mathbf{y}_{k}^{T_{k}}$ that causes erroneous decoding of the $k$-th message when $G^k=g^k$, i.e., 
\begin{align}
A_{k}(g^{k}):=\Big\{\mathbf{y}_{k}^{T_{k}}\in \mathcal{Y}^{Tn}\big|\psi_{k}^*\big(g^{k-1},\mathbf{y}^{T_{k}}_{k}\big)\neq g_{k}\Big\}.
\end{align}
\newpage
Then, the error probability of the $k$-th message can be written as follows: 
\begin{align}
&\Pr(\hat{G}_{k}\ne G_{k})\cr
&= \sum_{g^{k}\in \mathcal{G}^k}\frac{1}{2^{knR}}\Pr\left(\mathbf{Y}_{k}^{T_{k}}\in A_k(g^{k})|G^{k}=g^{k}\right)\\
&=\sum_{g^{\min\{T_{k},S_n\}}\in \mathcal{G}^{\min\{T_{k},S_n\}}}\frac{1}{2^{\min\{T_{k},S_n\}nR}}\Pr\left(\mathbf{Y}_{k}^{T_{k}}\in A_k(g^{k})|G^{\min\{T_{k},S_n\}}=g^{\min\{T_{k},S_n\}}\right)\\
&=\sum_{g^{\min\{T_{k},S_n\}}\in \mathcal{G}^{\min\{T_{k},S_n\}}}\frac{1}{2^{\min\{T_{k},S_n\}nR}}\Pr\Big(\mathbf{Y}_{k}^{T_{k}}\in A_k(g^{k})|\mathbf{X}^{T_{k}}_{k}=\phi^{T_{k}}_{k}(g^{\min\{T_{k},S_n\}}),\cr
&\qquad\qquad\qquad\qquad\qquad\qquad\qquad\qquad\qquad\qquad\qquad G^{\min\{T_{k},S_n\}}=g^{\min\{T_{k},S_n\}}\Big)\\
&\overset{(a)}=\sum_{g^{\min\{T_{k},S_n\}}\in \mathcal{G}^{\min\{T_{k},S_n\}}}\frac{1}{2^{\min\{T_{k},S_n\}nR}}\Pr\left(\mathbf{Y}_{k}^{T_{k}}\in A_k(g^{k})|\mathbf{X}^{T_{k}}_{k}=\phi^{T_{k}}_{k}(g^{\min\{T_{k},S_n\}})\right)\\
&=\sum_{g^{\min\{T_{k},S_n\}}\in \mathcal{G}^{\min\{T_{k},S_n\}}}\frac{1}{2^{\min\{T_{k},S_n\}nR}}W^{Tn}(A_k(g^{k})|\phi^{T_{k}}_{k}(g^{\min\{T_{k},S_n\}})), \label{eqn:eq_err}
\end{align}
where $\phi^{T_{k}}_{k}(g^{\min\{T_{k},S_n\}}):=(\phi_j(g^{\min\{j,S_n\}}): j\in [k:T_k])$\footnote{We remind that $\phi_j$ denotes the encoding function for the $j$-th block. } and $(a)$ is due to the Markov chain $\mathbf{Y}^{T_{k}}_{k}-\mathbf{X}^{T_{k}}_{k}-G^{\min\{T_{k},S_n\}}$.

\subsection{Lower bounding the error probability under an auxiliary channel}\label{subsec:lb}
In this subsection, we lower bound the maximal error probability over the first $S_n^*$ messages under an auxiliary channel $V_n^*$, where\footnote{Note that $(1-3(t+\zeta))/4>0$.} 
\begin{align}
S_n^*&:=\min\big\{S_n,  \exp\{n^{(1-3(t+\zeta))/4}\}\big\}.
\end{align} 
Note that the decoder decodes a total of $S_n^*$ messages (a total of $S_n^*nR_n$ bits) in $T_{S_n^*}$ blocks ($T_{S_n^*}n$ channel uses), which yields an effective rate of $\frac{S_n^*nR_n}{T_{S_n^*}n}=R_n-\frac{(T-1)R_n}{T_{S_n^*}}=: R_n-\delta_n$. Because $S_n=\omega(n^t)$ and $\exp\{n^{(1-3(t+\zeta))/4}\}=\omega(n^t)\cap \exp\{o(n^{(1-3(t+\zeta))/2})\}$, it follows that
\begin{align}
S_n^*&=\omega(n^t) \cap  \exp\{o(n^{(1-3(t+\zeta))/2})\},
\end{align}
and, in turn, 
\begin{align}
\delta_n=o(n^{-t})\cap  \exp\{-o(n^{(1-3(t+\zeta))/2})\}. \label{eqn:delta_n}
\end{align}
Now, we choose the auxiliary channel $V_n^*$ that optimizes the Haroutunian error exponent at rate $R_n-2\delta_n$, i.e.,  
\begin{equation}
V_n^* := \argmin_{V : C(V) \le R_n- 2\delta_n } \max_{P\in \mathcal{P}(\mathcal{X})} D(V\| W| P).  \label{eqn:min_V}
\end{equation}

The following lemma gives a lower bound on the maximal error probability over $S_n^*$ messages under the auxiliary channel $V_n^*$ using the fact that the effective rate $R_n-\delta_n$ is strictly larger than the capacity $R_n-2\delta_n$. This lemma is proved at the end of this section. 
\begin{lemma}\label{lemma:lb} 
Assume that the streaming code with the sequence of decoding functions  $\psi_k^*$ for $k\in [1:S_n]$  is applied to the auxiliary channel $V_n^*$. Then, there exists $\delta_n'= \Theta\left(\frac{\delta_n}{-\log \delta_n}\right)$ such that 
\begin{equation}
\max_{k\in [1:S_n^*]} \Pr(\hatG_{k}\neq G_{k})  \ge \delta'_n.\label{eqn:max_error}
\end{equation}
\end{lemma}

Let $k^*$ denote the message index whose error probability is the same as the maximal error probability over $S_n^*$ messages under the auxiliary channel $V_n^*$. In the subsequent subsection, we use the following corollary of Lemma~\ref{lemma:lb}, which is proved at the end of this section.
\begin{corollary}\label{corollary:proportion}
For at least a $\frac{\delta'_n}{2}$ proportion of sequences $g^{\min\{T_{k^*},S_n\}}$ in $\mathcal{G}^{\min\{T_{k^*},S_n\}}$, it follows that 
\begin{align}
(V_n^*)^{Tn}\left(A_{k^*}(g^{k^*})|\phi^{T_{k^*}}_{k^*}(g^{\min\{T_{k^*},S_n\}})\right)\geq \frac{\delta'_n}{2}.
\end{align}
 
\end{corollary}

\subsection{Change-of-measure} \label{subsec:cm}
Now, we lower bound the error probability of the $k^*$-th message under the true channel $W$ using the result in Corollary \ref{corollary:proportion}. To that end, we use the following lemma concerning a change-of-measure from the auxiliary channel $V_n^*$ to the true channel $W$. This lemma is particularly suited to moderate deviations analysis. The proof of this lemma is given in Appendix~\ref{appendix:cm}.
\begin{lemma}[Change-of-measure] \label{lemma:ch_m} If $(V_n^*)^{Tn}(A|x^{Tn})\geq \frac{\delta'_n}{2}$ for some $x^{Tn}\in \mathcal{X}^{Tn}$ and  $A\subseteq \mathcal{Y}^{Tn}$, the conditional probability under the true channel $W$ is lower-bounded as  
\begin{align}
W^{Tn}(A|x^{Tn})&\geq \frac{\delta'_n}{4}\exp\Big\{-Tn \Big( D( V_n^* \|W|P_{x^{Tn}}) + \eta_n\Big) \Big\} \label{eqn:cm} 
\end{align}
for some $\eta_n=o(n^{-2t})$. We note that the condition $t<\frac{1}{3}$ for the moderate deviations scaling is crucial in the derivation of \eqref{eqn:cm}.
\end{lemma}

Now, we have 
\begin{align}
\Pr(\hat{G}_{k^*}\ne G_{k^*})
&\overset{(a)}=\sum_{g^{\min\{T_{k^*},S_n\}}}\frac{1}{2^{\min\{T_{k^*},S_n\}nR}}W^{Tn}(A_{k^*}(g^{k^*})|\phi^{T_{k^*}}_{k^*}(g^{\min\{T_{k^*},S_n\}}))\\
&\overset{(b)}\geq  \frac{(\delta'_n)^2}{8}\exp\Big\{-Tn \Big( \max_{P \in\calP(\calX)} D( V_n^* \|W|P) +  o(n^{-2t})\Big) \Big\}\\
&= \frac{(\delta'_n)^2}{8}\exp\Big\{-Tn \big(E^+(R_n-2\delta_n) + o(n^{-2t})\big) \Big\}\\
&\overset{(c)}= \frac{(\delta'_n)^2}{8}\exp\Big\{-Tn \big(E_{\mathrm{SP}}(R_n-2\delta_n) + o(n^{-2t})\big) \Big\}, \label{eqn:last_lb}
\end{align}
where $(a)$ is from \eqref{eqn:eq_err}, $(b)$ is due to Corollary~\ref{corollary:proportion} and Lemma~\ref{lemma:ch_m}, and $(c)$ is because $W$ is assumed to be output symmetric \cite{CsiszarKorner:11, Haroutunian:77}. 
Because $\delta_n'= \Theta\left(\frac{\delta_n}{-\log \delta_n}\right)$ and  $\delta_n$  satisfies \eqref{eqn:delta_n}, it follows that 
\begin{align}
\epsilon_n&\geq  \frac{\exp\{-o(n^{(1-3(t+\zeta))/2})\}}{o(n^{1-3(t+\zeta)})} \exp\Big\{-Tn \big(E_{\mathrm{SP}}(C-n^{-t}-o(n^{-t})) + o(n^{-2t})\big) \Big\}.
\end{align}
By taking the logarithm and normalizing by $-n^{1-2t}$, we obtain
\begin{align}
-\frac{1}{n^{1-2t}}\log \epsilon_n&\leq \frac{o(n^{(1-3(t+\zeta))/2}) +\log(o(n^{1-3(t+\zeta)}))}{n^{1-2t}}+Tn^{2t}\big(E_{\mathrm{SP}}(C-n^{-t}-o(n^{-t})) + o(n^{-2t})\big) .\label{eqn:loga}
\end{align}
To asymptotically bound the term involving $E_{\mathrm{SP}}(\cdot)$, we use the following lemma.
\begin{lemma}[{\cite[Proposition 1]{AltugWagner:14}}] \label{lemma:sp} When $\rho_n>0$ satisfies $\rho_n\rightarrow 0$ and $\rho_n\sqrt{n}\rightarrow \infty$, 
\begin{align}
\limsup_{n\rightarrow \infty} \frac{E_{\mathrm{SP}}(C-\rho_n)}{\rho_n^2}\leq \frac{1}{2\nu}.
\end{align}
\end{lemma}
Now, by taking limit superior to both sides of \eqref{eqn:loga} and applying Lemma~\ref{lemma:sp}, we obtain
\begin{align}
\limsup_{n\rightarrow \infty}-\frac{1}{n^{1-2t}}\log \epsilon_n&\leq \frac{T}{2\nu},
\end{align}
which completes the proof.
\end{proof}
\begin{remark}\label{remark:novelty}
The main flow of our converse proof is similar with that in \cite[Section IV]{Sahai:08} which is for a bit-wise streaming setup in the large deviations regime. In the following, the main technical novelty in our converse proof is summarized. 
\begin{itemize}
\item In \cite[Section IV]{Sahai:08}, the term corresponding to $S_n^*$ is a constant independent of $n$ and thus the resultant terms corresponding to $\delta_n$ and $\delta_n'$ are also constants.\footnote{We note that $\delta_n$ and $\delta_n'$ can be determined from $S_n^*$. } In our proof, $S_n^*$ is chosen carefully to simultaneously ensure that (i) the backoff from capacity is not affected by the subtraction of $2\delta_n$  (in e.g., \eqref{eqn:last_lb}), (ii) the moderate deviations constant is not affected by the multiplicative term $\frac{\delta_n'^2}{8}$ (in e.g., \eqref{eqn:last_lb}), and (iii) in the proof of the change-of-measure lemma, the speed of  convergence of the probability of a typical set to unity is asymptotically higher than the speed of decay of $\delta_n'$  (i.e., to derive Eq. \eqref{eqn:pb_joint}).

\item In the change-of-measure lemma \cite[Lemma 4.4]{Sahai:08}, the remainder term  corresponding to $\eta_n$ in Lemma~\ref{lemma:ch_m} is a constant independent of $n$. In the moderate deviations regime, the error probability decays subexponentially and hence it should be proved that the remainder term decays to zero sufficiently fast so that it does not affect the moderate deviations constant. In Lemma~\ref{lemma:ch_m}, which corresponds to a change-of-measure lemma suited to moderate deviations analysis, the speed of decay of $\eta_n$ is asymptotically bounded by carefully choosing the parameters of a typical set and characterizing the asymptotic behavior of the ratio of $V_n^*$ to $W$ (i.e., Lemma~\ref{lem:cont}). 

\item The work \cite[Section IV]{Sahai:08} analyzes the bit-wise error under the bit-wise encoding and decoding operations. In the proof of Lemma~\ref{lemma:lb}, we develop proof techniques adapted to the message-wise error under the block-wise operations. 
\end{itemize}
\end{remark}
\begin{remark}
We note that the condition $t<\frac{1}{3}$ for the moderate deviations scaling  in Proposition \ref{prop:converse} is not needed in the proof steps preceding Lemma~\ref{lemma:ch_m}. In the proof of Lemma~\ref{lemma:ch_m}, we need the condition $t<\frac{1}{3}$ to make the parameters of a typical set simultaneously satisfy that (i) the probability of typical set converges to unity as the length of the sequences increases and (ii) the remainder term $\eta_n$ is $o(n^{-2t})$. 
\end{remark} 

\begin{proof}[Proof of Lemma~\ref{lemma:ff_recent}]
The proof is immediate from the following Markov chain: 
\begin{align}
\mathbf{Y}^{k-1}-(G^{k-1},\mathbf{Y}_{k}^{T_k})-G_k,\label{eqn:markov}
\end{align}
which holds due to the causal nature of the encoder and the memoryless nature of the channel.

To prove explicitly, let $\mathcal{\psi}_{k,\mathrm{map}}^{f}$ be the maximum a posteriori probability (MAP) decoding function for message $G_k$ based on feedforward information $G^{k-1}$ and channel output sequences $\mathbf{Y}^{T_k}$. Then, we have
\begin{align}
\mathcal{\psi}_{k,\mathrm{map}}^{f}(g^{k-1},\mathbf{y}^{T_k})&=\argmax_{g_k}\Pr(G_k=g_k|g^{k-1},\mathbf{y}^{T_k})  \\
 &\overset{(a)}=\argmax_{g_k}\Pr(G_k=g_k|g^{k-1},\mathbf{y}^{T_k}_k) ,
\end{align}
where $(a)$ is due to the Markov chain in \eqref{eqn:markov}. 

Now, let us define $\psi_k^*$ for $k\in [1:S_n]$ as follows: 
\begin{align}
\psi_k^*(g^{k-1}, \mathbf{y}^{T_k}_k)=\argmax_{g_k}\Pr(G_k=g_k|g^{k-1},\mathbf{y}^{T_k}_k),
\end{align}
which achieves the same performance as $\mathcal{\psi}_{k,\mathrm{map}}^{f}$. Because MAP decoding is optimal, the probability of error using $\psi_k^*$ is the same as or less than that using any feedforward decoding function $\psi_k^f$. 
\end{proof}

\begin{proof}[Proof of Lemma~\ref{lemma:lb}]
In this proof, all the entropy and mutual information terms and probabilities are evaluated under the input distribution induced by the assumed streaming code with the sequence of decoding functions  $\psi_k^*$ for $k\in [1:S_n]$  and the auxiliary channel $V_n^*$. Let $\hat{G}_k$ for $k\in[1:S_n]$ denote the output of $\psi_k^*$. Let us denote  the binary representations of $G_k$ and $\hat{G}_k$ for $k\in [1:S_n]$ by $\mathbf{B}_k=B_{(k-1)nR_n+1}^{knR_n}$ and $\hat{\mathbf{B}}_k=\hat{B}_{(k-1)nR_n+1}^{knR_n}$, respectively, where the $B_i$'s and $\hat{B}_i$'s are binary, i.e., in $\{0,1\}$.  Let $\tilde{\mathbf{B}}_k:=\mathbf{B}_k \oplus \hat{\mathbf{B}}_k $ denote the error sequence of the $k$-th message. 
Then, we have
\begin{align}
S_n^*nR_n&= H(\mathbf{B}^{S_n^*})\\
&=I(\mathbf{B}^{S_n^*};\mathbf{B}^{S_n^*})\\
&\leq I(\mathbf{B}^{S_n^*};\mathbf{B}^{S_n^*},\tilde{\mathbf{B}}^{S_n^*},\mathbf{Y}^{T_{S_n^*}})\\
&\overset{(a)}=I(\mathbf{B}^{S_n^*};\tilde{\mathbf{B}}^{S_n^*},\mathbf{Y}^{T_{S_n^*}})\\
&=I(\mathbf{B}^{S_n^*};\mathbf{Y}^{T_{S_n^*}})+I(\mathbf{B}^{S_n^*};\tilde{\mathbf{B}}^{S_n^*}|\mathbf{Y}^{T_{S_n^*}})\\
&\leq I(\mathbf{B}^{S_n^*};\mathbf{Y}^{T_{S_n^*}})+H(\tilde{\mathbf{B}}^{S_n^*})\\
&\leq I(\mathbf{B}^{S_n^*},\mathbf{X}^{T_{S_n^*}};\mathbf{Y}^{T_{S_n^*}})+H(\tilde{\mathbf{B}}^{S_n^*})\\
&\overset{(b)}= I(\mathbf{X}^{T_{S_n^*}};\mathbf{Y}^{T_{S_n^*}})+H(\tilde{\mathbf{B}}^{S_n^*}), \label{eqn:block}
\end{align}
where $(a)$ is because $\mathbf{B}^{S_n^*}$ can be reconstructed from $(\tilde{\mathbf{B}}^{S_n^*},\mathbf{Y}^{T_{S_n^*}})$, which is proved at the end of this proof,  and $(b)$ is due to the memoryless nature of the channel.

Since the capacity of $V_n^*$ does not exceed $R_n-2\delta_n$, it follows that 
\begin{align}
I(\mathbf{X}^{T_{S_n^*}};\mathbf{Y}^{T_{S_n^*}})\leq T_{S_n^*}n(R_n-2\delta_n). \label{eqn:cap_cond}
\end{align}
By combining \eqref{eqn:block} and \eqref{eqn:cap_cond}, we obtain 
\begin{align}
H(\tilde{\mathbf{B}}^{S_n^*})&\geq S_n^*nR_n-T_{S_n^*}n(R_n-2\delta_n) \\
&\overset{(a)}=T_{S_n^*}n\delta_n
\end{align}
where $(a)$ is because $S_n^*nR_n=T_{S_n^*}n(R_n-\delta_n)$.

Since the average of the marginal entropy terms satisfies 
\begin{align}
\frac{1}{T_{S_n^*}n(R_n-\delta_n)}\sum_{i=1}^{T_{S_n^*}n(R_n-\delta_n)}H(\tilde{B}_i)\geq \frac{\delta_n}{R_n-\delta_n},
\end{align}
there exists $i'\in [1:S_n^*nR_n]$ such that 
\begin{align}
H(\tilde{B}_{i'})\geq  \frac{\delta_n}{R_n-\delta_n}. 
\end{align}
Then, by the fact that the binary entropy function satisfies $\lim_{p\to 0} \frac{h(p)}{ -p\log p} = 1$,  there exists $\delta'_n = \Theta\left(\frac{\delta_n}{-\log \delta_n}\right)$ such that 
\begin{equation}
\Pr(\tilB_{i'}=1) =\Pr(\hatB_{i'}\ne B_{i'}) \ge \delta'_n.\label{eqn:perr_error_seq}
\end{equation}
This in turn implies that there exists $k'\in [1:S_n^*]$ such that 
\begin{equation}
\Pr(\hatG_{k'}\ne G_{k'}) \ge \delta'_n, 
\end{equation}
and hence 
\begin{equation}
\max_{k\in[1:S_n^*]}\Pr(\hatG_{k}\ne G_{k}) \ge \delta'_n.
\end{equation}

Now, it remains to show that $\mathbf{B}^{S_n^*}$ can be reconstructed from $(\tilde{\mathbf{B}}^{S_n^*},\mathbf{Y}^{T_{S_n^*}})$. To that end, let us show that there exists  a sequence of  functions $f_k: \{0,1\}^{knR} \times \mathcal{Y}^{T_kn} \rightarrow \{0,1\}^{knR} $ for $k\in [1:S_n]$ such  that  
\begin{align}
f_k(\tilde{\mathbf{B}}^{k}, \mathbf{Y}^{T_k}) = \mathbf{B}^{k}.\label{eqn:error}
\end{align} 
This can be proved by using induction. For $k=1$, assume that $(\tilde{\mathbf{B}}_1, \mathbf{Y}^{T_1})$ is given. Then, $\hat{\mathbf{B}}_1$ can be obtained by representing $\psi_1^*(\mathbf{Y}^{T_1})$ in binary, and in turn, $\mathbf{B}_1$ can be reconstructed by XOR-ing $\hat{\mathbf{B}}_1$ with $\tilde{\mathbf{B}}_1$. Hence, there exists $f_1$ that satisfies \eqref{eqn:error} for $k=1$. 

Now, fix $k\geq 2$ and assume that $(\tilde{\mathbf{B}}^{k}, \mathbf{Y}^{T_{k}})$ is given. Assume that there exists $f_{k-1}$  such that
\begin{align}
f_{k-1}(\tilde{\mathbf{B}}^{k-1}, \mathbf{Y}^{T_{k-1}})= \mathbf{B}^{k-1}.\label{eqn:err_induct}
\end{align}
Then, $\mathbf{B}^{k-1}$ and $G^{k-1}$ can be obtained from $f_{k-1}$. Furthermore, $\hat{\mathbf{B}}_k$ can be obtained by representing $\psi^*_k(G^{k-1},\mathbf{Y}_k^{T_k})$ in binary, and thus $\mathbf{B}_k$ can be reconstructed by XOR-ing $\hat{\mathbf{B}}_k$ with $\tilde{\mathbf{B}}_k$. Hence, $\mathbf{B}^{k}$ can be obtained from  $(\tilde{\mathbf{B}}^{k}, \mathbf{Y}^{T_{k}})$ so there exists $f_k$ that satisfies \eqref{eqn:error}. 
\end{proof} 

\begin{proof}[Proof of Corollary \ref{corollary:proportion}]
Let $\mu:=\min\{T_{k^*},S_n\}$ for notational convenience. We have 
\begin{align}
\delta'_n&\overset{(a)}\leq \Pr(\hat{G}_{k^*}\ne G_{k^*})\\
&\overset{(b)}=\sum_{g^{\mu}\in\mathcal{G}^{\mu}}\frac{1}{2^{\mu nR}}(V_n^*)^{Tn}\left(A_{k^*}(g^{k^*})|\phi^{T_{k^*}}_{k^*}(g^{\mu})\right),
\end{align}
where $(a)$ is due to Lemma~\ref{lemma:lb} and $(b)$ is from the same chain of equalities used to obtain \eqref{eqn:eq_err} by assuming the auxiliary channel $V_n^*$ instead of $W$. 
Now, let us assume, to the contrary, that for strictly less than $\frac{\delta'_n}{2}$ proportion of sequences $g^{\mu}$ in $\mathcal{G}^{\mu}$, the conditional probability $(V_n^*)^{Tn}\left(A_{k^*}(g^{k^*})|\phi^{T_{k^*}}_{k^*}(g^{\mu})\right)$ is at least $\frac{\delta'_n}{2}$. Let $B$ denote this subset of $\mathcal{G}^{\mu}$. Then,  we have
\begin{align}
&\sum_{g^{\mu}\in\mathcal{G}^{\mu}}\frac{1}{2^{\mu nR}}
(V_n^*)^{Tn}\left(A_{k^*}(g^{k^*})|\phi^{T_{k^*}}_{k^*}(g^{\mu})\right)\cr
&= \sum_{g^{\mu}\in B}\frac{1}{2^{\mu nR}}
(V_n^*)^{Tn}\left(A_{k^*}(g^{k^*})|\phi^{T_{k^*}}_{k^*}(g^{\mu})\right)
+\sum_{g^{\mu}\in B^c}\frac{1}{2^{\mu nR}}
(V_n^*)^{Tn}\left(A_{k^*}(g^{k^*})|\phi^{T_{k^*}}_{k^*}(g^{\mu})\right)\\
&\overset{(a)}\leq \frac{1}{2^{\mu nR}}\cdot |B|\cdot 1+ \frac{1}{2^{\mu nR}}\cdot|B^c| \cdot \frac{\delta'_n}{2}\\
&\overset{(b)}<  \frac{\delta'_n}{2} +  \frac{\delta'_n}{2}\\
&=\delta_n',
\end{align}
where $(a)$ is obtained by upper bounding the conditional probabilities by 1 and $\delta_n'/2$ for the sequences in $B$ and $B^c$, respectively, and $(b)$ is because $|B|<\frac{\delta_n'}{2}\cdot |\mathcal{G}^{\mu}|= \frac{\delta_n'}{2}\cdot 2^{\mu nR}$ and $|B^c|\leq  |\mathcal{G}^{\mu}|=   2^{\mu nR}$. This is a contradiction and hence the proof is completed. 
\end{proof}

\section{Achievability} \label{sec:achievability}
\begin{proof}[Proof of Proposition \ref{prop:achievability}]
Consider a DMC $(\mathcal{X},\mathcal{Y}, \{W(y|x): x\in \mathcal{X}, y\in \mathcal{Y}\} )$ with $\nu>0$. We show that there exists a sequence of $(n,M_n,\epsilon_n,T,S_n)$-streaming codes such that  $\log M_n=nC-n^{1-t}$, $S_n=\exp\{o(n^{1-2t})\}$, and $\epsilon_n$ satisfies \eqref{eqn:MD_achiev} for $0<t<\frac{1}{2}$. 
The encoding and decoding procedures are the same as those in \cite[Section IV-A]{LeeTanKhisti:arxiv15}, which are summarized in the following for the sake of completeness. We borrow some of the steps in the error analysis from \cite[Section IV-A]{LeeTanKhisti:arxiv15} as well, but the main difference is in considering the maximal error probability rather than the average error probability. 

Let $P_X$ denote an input distribution that achieves the dispersion. For the sake of symmetry in describing the encoding and decoding procedures, in addition to the sequence of messages $\{G_k\}_{k\in [1:S_n]}$, we introduce a sequence of auxiliary messages $\{G_k\}_{k\in [S_n+1:T_{S_n}]}$ each distributed uniformly over $\mathcal{G}$ that do not need to be decoded. 

\subsubsection{Encoding} 
For each $k\in [1:T_{S_n}]$ and $g^{k} \in \mathcal{G}^{k}$, generate $\mathbf{x}_k(g^{k})$ in an i.i.d. manner according to $P_X$. The generated codewords constitute the codebook $\mathcal{C}_n$. In block $k$, after observing the message sequence $G^{k}$, the encoder sends $\mathbf{x}_k(G^{k})$. 

\subsubsection{Decoding} 
Consider the decoding of $G_k$ at the end of block $T_k$ for $k\in [1:S_n]$. The decoder not only decodes $G_k$, but also re-decodes $G_1,\cdots, G_{k-1}$ at the end of block $T_k$. Let $\hat{G}_{T_k,j}$ denote the estimate of $G_j$ at the end of block $T_k$. The decoder decodes $G_j$ sequentially from $j=1$ to $j=k$ as follows: 
\begin{itemize}
\item Given $\hat{G}_{T_k,[1:j-1]}$, the decoder chooses $\hat{G}_{T_k,j}$ according to the following rule.\footnote{When $j=1$, $\hat{G}^{j-1}_{T_k}$ is null. } If there is a unique index $g_j\in \mathcal{G}$ that satisfies\footnote{The following notation is used for the set of codewords. Let $\mathcal{K}_j$ for $j\in \bbN$ denote the set of message indices mapped to the $j$-th codeword according to the encoding procedure. For $\mathcal{J}\subseteq \bbN$ and $\mathcal{K}\supseteq \bigcup_{j\in \mathcal{J}} \mathcal{K}_j$, we denote by  $\mathbf{x}_{\mathcal{J}}(g_{\mathcal{K}})$ the set of codewords $\{\mathbf{x}_j(g_{\mathcal{K}_j}): j\in \mathcal{J}\}$. } 
\begin{align}
i(\mathbf{x}_{[j:T_k]}(\hat{G}_{T_k, [1:j-1]}, g_{j}^{T_k}), \mathbf{y}_{[j:T_k]})&> (T_k-j+1) \cdot \log M_n \label{eqn:dec_rule}
\end{align}
for some $g_{j+1}^{T_k}$, let $\hat{G}_{T_k,j}=g_j$. 
If there is none or more than one such $g_j$, let $\hat{G}_{T_k,j}=1$. 
\item If $j<k$, repeat the above procedure by increasing $j$ to $j+1$. If $j=k$, the decoding procedure terminates and the decoder declares that the $k$-th message is $\hat{G}_{k}:=\hat{G}_{T_k,k}$.
\end{itemize}

\subsubsection{Error analysis} The aforementioned encoding and decoding procedures are the same as in \cite[Section IV-A]{LeeTanKhisti:arxiv15}. Hence, due to the same analysis used to derive \cite[Eq. (30)]{LeeTanKhisti:arxiv15}, it follows that for arbitrary $0<\lambda<1$, 
\begin{align}
\E_{\mathcal{C}_n}[\Pr(\hat{G}_k\neq G_k|\mathcal{C}_n)]\leq \frac{\exp\left\{-Tn^{1-2t}\lambda^2\left(\frac{1 }{2\nu}-\lambda n^{-t} \tau\right)\right\}}{1-\exp\{-n^{1-2t}\lambda^2\left(\frac{1 }{2\nu}-\lambda n^{-t} \tau\right)\}}+\frac{\exp\left\{-Tn^{1-t}(1-\lambda)\right\}}{1-\exp\left\{-n^{1-t}(1-\lambda)\right\}} \label{eqn:MD_sumup}
\end{align}
for sufficiently large $n$, where $\tau$ is some non-negative constant dependent only on the input distribution $P_X$ and the channel $W$. Then, we have 
\begin{align}
\E_{\mathcal{C}_n}\Big[\max_{k\in[1:S_n]} \Pr(\hat{G}_k\neq G_k|\mathcal{C}_n)\Big]&\leq \E_{\mathcal{C}_n}\Big[\sum_{k\in[1:S_n]} \Pr(\hat{G}_k\neq G_k|\mathcal{C}_n)\Big]\\
&\leq S_n\frac{\exp\left\{-Tn^{1-2t}\lambda^2\left(\frac{1 }{2\nu}-\lambda n^{-t} \tau\right)\right\}}{1-\exp\{-n^{1-2t}\lambda^2\left(\frac{1 }{2\nu}-\lambda n^{-t} \tau\right)\}}+S_n\frac{\exp\left\{-Tn^{1-t}(1-\lambda)\right\}}{1-\exp\left\{-n^{1-t}(1-\lambda)\right\}} 
\end{align}
for sufficiently large $n$. Because $S_n=\exp\{o(n^{1-2t})\}$ , we obtain
\begin{align}
\liminf_{n\rightarrow \infty} -\frac{1}{n^{1-2t} }\log \E_{\mathcal{C}_n}\left[\max_{k\in[1:S_n]} \Pr(\hat{G}_k\neq G_k|\mathcal{C}_n)\right]\geq \frac{T\lambda^2}{2\nu}.
\end{align}
By taking $\lambda\rightarrow 1$, we have 
\begin{align}
\liminf_{n\rightarrow \infty} -\frac{1}{n^{1-2t} }\log \E_{\mathcal{C}_n}\left[\max_{k\in[1:S_n]} \Pr(\hat{G}_k\neq G_k|\mathcal{C}_n)\right]\geq \frac{T}{2\nu}.
\end{align}
Hence, there must exist a sequence of codes $\mathcal{C}_n$ that satisfies  \eqref{eqn:MD_achiev}, which completes the proof. 
\end{proof}

\section{Conclusion} \label{sec:conclusion}
In this paper, we studied the moderate deviation asymptotics  for a streaming setup with a decoding delay of $T$ blocks. We showed that  the moderate deviations constant for output symmetric channels improves over the block coding or non-streaming setup exactly by a factor of $T$ for a certain range of moderate deviations scalings under some mild conditions on the number of streaming messages. 

We note that our converse result holds only for output symmetric channels because the Haroutunian exponent is the same as the sphere packing exponent for such channels.
The output symmetry of the channel would not be necessary if the Haroutunian exponent behaves as the sphere packing exponent in the moderate deviations regime for general DMCs, i.e., $E^+(C-n^{-t})\approx \frac{n^{-2t}}{2\nu}$ (compare to Lemma~\ref{lemma:sp}), which does not seem obvious since the Haroutunian exponent is strictly greater than the sphere packing exponent for some asymmetric channels. 
Alternatively, one could attempt to derive the sphere-packing bound directly as in \cite{AltugWagner:14} for block channel coding. In this approach, we first assume by using some standard arguments that the channel input sequences over $T_{S_n^*}$ blocks are constant composition, say type $P$, and then choose the auxiliary channel $V_n^*$ as follows instead of \eqref{eqn:min_V}:\footnote{We remind that $R_n=C-n^{-t}$ and $\delta_n=o(n^{-t})$.} 
\begin{align}
V_n^*=\argmin_{V: I(P,V)\leq R_n-2\delta_n} D(V\|W|P).
\end{align}
Then, due to similar arguments as in Section \ref{subsec:lb}, there exists a message index $k^*$ whose error probability is at least $\delta_n'=\Theta(\frac{\delta_n}{-\log \delta_n})$ under the channel $V_n^*$. 
As also remarked in \cite[Section IV-D]{Sahai:08}, the problem of this approach arises in the change-of-measure step since the dominant type of the channel input sequences in the $T$ blocks from the $k^*$-th block, i.e., block $k^*$ to block $T_{k^*}$, may not be the same as $P$.

Finally, we discuss whether it is possible to generalize the techniques herein to the case where the channel is Gaussian and there is a peak (almost sure) power constraint on the codewords. Close inspection of the upper bound~\eqref{eqn:MD_sumup} on the error probability in the achievability proof  together with a standard change-of-measure technique (e.g.,~\cite{MolavianJaziLaneman:15}), allows us to conclude that the achievability bound in Proposition~\ref{prop:achievability} continues to hold with $\nu = \frac{P(P+2)\log^2 e}{2(P+1)^2}$ (assuming we use bits as the units of information), where $P$ is the peak power of the codewords. However, the converse is not straightforward as the proof in Section~\ref{sec:converse} hinges on the use of the method of types and an analogue of strong typicality (cf.\ Lemma~\ref{lemma:typical}). These tools are more suited to channels with finite alphabets and cannot be easily adapted to channels with uncountable alphabets such as Gaussian channels. Thus, it appears that some new techniques are required to establish the analogue of Proposition~\ref{prop:converse} for Gaussian channels. However, we note that the converse proof herein uses several analytical tools that are used to analyze DMCs with feedback (e.g., the Haroutunian exponent). For Gaussian channels under the peak power constraint, there are some recent works~\cite{PolyanskiyPoorVerdu:11, FongTan:15}  that show that feedback does not improve the second- and third-order performance. Thus, the analytical tools in these recent works may pave a way to establish a result similar to that in Proposition~\ref{prop:converse}.

\appendices

\section{Change-of-measure}\label{appendix:cm}
\begin{proof}[Proof of Lemma~\ref{lemma:ch_m}]
To prove Lemma~\ref{lemma:ch_m}, let us first present two lemmas. First, the following lemma is a refined version of the typical set lemma \cite[Lemma II.1]{Sahai:08} and is proved at the end of this appendix.

\begin{lemma} \label{lemma:typical}
Consider a DMC $(\mathcal{X}, \mathcal{Y}, \{V(y|x): x\in \mathcal{X}, y\in \mathcal{Y}\})$, a vector $x^l\in \mathcal{X}^l$ of length $l$, $\gamma_{1}>0$ and $\gamma_{2}>0$. 
The following holds: 
\begin{align}
V^l(\calJ^{\gamma_{1},\gamma_{2}}_{x^l}|x^l)\geq 1-2|\mathcal{X}||\mathcal{Y}|\exp\{-2\gamma_{1}^2\gamma_{2}l\}, \label{eqn:typical}
\end{align}
where  the  typical set $\calJ^{\gamma_{1},\gamma_{2}}_{x^l}$ is defined as follows:
\begin{align}
\calJ^{\gamma_{1},\gamma_{2}}_{x^l}:=&\Big\{y^l\in \mathcal{Y}^l\Big| \mbox{ for every } x\in \mathcal{X}  \mbox{ such that } \frac{N_x(x^l)}{l}\geq \gamma_2 \mbox{ and every } y\in \mathcal{Y}, \cr
&\qquad\qquad\qquad \qquad\qquad\qquad ~ \Big|\frac{N_{x,y}(x^l,y^l)}{N_x(x^l)}-V(y|x)\Big|<\gamma_1 \Big\}.
\end{align}

Furthermore, for any $y^l \in \calJ_{x^l}^{\gamma_1,\gamma_2}$ such that $V^l(y^l|x^l)\neq 0$,
 \begin{equation}
\frac{W^l(y^l|x^l)}{V^l(y^l | x^l)}\ge \exp\Big\{-l\big( D(V\|W|P_{x^l}) +(\gamma_1+2 \gamma_2)\gamma'\big)\Big\} , \label{eqn:tilting}
 \end{equation}
where $\gamma':=\sum_{(x,y): V(y|x)\ne 0  } \left| \log\frac{V(y|x)}{W(y|x)} \right|$. 
\end{lemma}

The following lemma states that $\gamma'$ vanishes to zero sufficiently fast if we apply Lemma~\ref{lemma:typical} with the substitution of $V\Leftarrow V_n^*$. The proof of this lemma is relegated to the end of this appendix. 
\begin{lemma} \label{lem:cont}
Let  $\gamma_n':=\sum_{(x,y):V_n^*(y|x)\ne 0 } \left| \log\frac{V_n^*(y|x)}{W(y|x)} \right|$. Then, $\gamma_n'=O(n^{-t})$.
\end{lemma}

Now we are ready to prove Lemma~\ref{lemma:ch_m}. Fix $x^{Tn}\in \mathcal{X}^{Tn}$ and $A\subseteq \mathcal{Y}^{Tn}$. 
We apply Lemma~\ref{lemma:typical} with the substitution of $V\Leftarrow V_n^*$ and $l\Leftarrow Tn$. To make the typical set satisfy the usual property that its probability approaches unity as the length of the sequences tends to infinity, we choose  
\begin{equation}
\gamma_{1,n} = n^{-(t+\zeta)}, ~\gamma_{2,n}=n^{-(t+\zeta)}, \label{eqn:def_eps1}
\end{equation}  
where $\zeta=\frac{1}{2}(\frac{1}{3}-t)>0$. Then we have 
  \begin{equation}
(V_n^*)^{Tn} (\calJ_{x^{Tn}}^{\gamma_{1,n},\gamma_{2,n}}  | x^{Tn} ) \ge 1-2|\mathcal{X}||\mathcal{Y}| \exp\{ - 2Tn^{\left(1-3(t+\zeta)\right)} \}=:1-\varphi_n\rightarrow 1,
  \end{equation}
because $1-3(t+\zeta)>0$. We note that because $\delta_n=\exp\{-o(n^{(1-3(t+\zeta))/2})\}$, it follows that 
\begin{align}
\delta_n'= \Theta\left(\frac{\delta_n}{-\log \delta_n}\right)=\frac{\exp\{-o(n^{(1-3(t+\zeta))/2})\}}{o(n^{(1-3(t+\zeta))/2})}.
\end{align}
Since $\varphi_n = o(\delta_n')$, we can find $n$ large enough so that $\varphi_n< \delta_n'/4$. Thus, by the union bound,
\begin{equation}
(V_n^*)^{Tn} (A\cap \calJ_{x^{Tn}}^{\gamma_{1,n},\gamma_{2,n}}  | x^{Tn} ) \ge \frac{\delta_n'}{4}\label{eqn:pb_joint}
\end{equation}
for sufficiently large $n$. 
Now, we obtain 
\begin{align}
W^{Tn}(A   |x^{Tn}) & \ge W^{Tn} (A\cap \calJ_{x^{Tn}}^{\gamma_{1,n},\gamma_{2,n}}  | x^{Tn} )  \\
& \overset{(a)}\ge  \frac{\delta_n'}{4} \exp \left\{ -Tn  \left( D( V_n^* \|W|P_{x^{Tn}}) + (\gamma_{1,n}+2\gamma_{2,n} ) \gamma_{n}' \right)\right\}\\
&\overset{(b)}= \frac{\delta_n'}{4} \exp \left\{ -Tn  \left( D( V_n^* \|W|P_{x^{Tn}}) + O(n^{-2t-\zeta})\right)\right\}
\end{align}
where $(a)$ is from Lemma~\ref{lemma:typical} and $(b)$ is due to the choice of $\gamma_{1,n}$ and $\gamma_{2,n}$ in \eqref{eqn:def_eps1} together with the asymptotic bound on $\gamma_n'$ in  Lemma~\ref{lem:cont}. Since $\zeta>0$, Lemma~\ref{lemma:ch_m} is proved. 
\end{proof}

\begin{proof}[Proof of Lemma~\ref{lemma:typical}]
Fix a DMC $(\mathcal{X}, \mathcal{Y}, \{V(y|x): x\in \mathcal{X}, y\in \mathcal{Y}\})$, a vector $x^l\in \mathcal{X}^l$ of length $l$, $\gamma_{1}>0$ and $\gamma_{2}>0$. 
Let $(\calJ^{\gamma_{1},\gamma_{2}}_{x^l})^c:=\mathcal{Y}^l\setminus \calJ^{\gamma_{1},\gamma_{2}}_{x^l}$. First, \eqref{eqn:typical} is proved as follows:
\begin{align}
V^l((\calJ^{\gamma_{1},\gamma_{2}}_{x^l})^c|x^l)&\leq \sum_{ x\in \mathcal{X}: \frac{N_x(x^l)}{l}\geq \gamma_2}\sum_{y\in \mathcal{Y}}\Pr\left( \Big|\frac{N_{x,y}(x^l,y^l)}{N_x(x^l)}-V(y|x)\Big|\geq\gamma_1 \right)\\
&\overset{(a)}\leq  \sum_{ x\in \mathcal{X}: \frac{N_x(x^l)}{l}\geq \gamma_2}\sum_{y\in \mathcal{Y}}2\exp\left\{-2\gamma_1^2 N_x(x^l)\right\}\\
&\leq 2|\mathcal{X}||\mathcal{Y}|\exp\{-2\gamma_{1}^2\gamma_{2}l\},
\end{align}
where $(a)$ is from the Chernoff bound. 

Next, fix any  $y^l \in \calJ_{x^l}^{\gamma_1,\gamma_2}$ such that $V^l(y^l|x^l)\neq 0$. Then,
\begin{align}
V^l(y^l | x^l)=\prod_{i=1}^l V(y_i|x_i)=\prod_{x\in\mathcal{X}, y\in \mathcal{Y}}V(y|x)^{N_{x,y}(x^l, y^l)}
\end{align}
and similarly 
\begin{align}
W^l(y^l | x^l)=\prod_{x\in\mathcal{X}, y\in \mathcal{Y}}W(y|x)^{N_{x,y}(x^l, y^l)}. 
\end{align}
Then, the ratio of the two probabilities is given as follows: 
\begin{align}
\frac{W^l(y^l | x^l)}{V^l(y^l | x^l)}&=\prod_{x\in\mathcal{X}, y\in \mathcal{Y}}\left(\frac{W(y|x)}{V(y|x)}\right)^{N_{x,y}(x^l, y^l)}\\
&=\exp\left\{-l\sum_{x\in\mathcal{X}, y\in \mathcal{Y}}\frac{N_{x,y}(x^l, y^l)}{l}\log\frac{V(y|x)}{W(y|x)} \right\}. \label{eqn:ratio}
\end{align}
To bound the summation term inside the exponential function in \eqref{eqn:ratio}, we consider the following two subsets of $\mathcal{X}$: 
\begin{align}
\mathcal{X}_1(x^l):=\left\{x: \frac{N_x(x^l)}{l}<\gamma_2\right\},~\mathcal{X}_2(x^l):=\mathcal{X}\setminus\mathcal{X}_1(x^l).
\end{align}
Then, we have
\begin{align}
&\sum_{x\in\mathcal{X}_1(x^l), y\in \mathcal{Y}}\frac{N_{x,y}(x^l, y^l)}{l}\log\frac{V(y|x)}{W(y|x)} \cr
&\overset{(a)}=\sum_{x\in\mathcal{X}_1(x^l), y\in \mathcal{Y}: V(y|x)\neq 0}\frac{N_{x,y}(x^l, y^l)}{l}\log\frac{V(y|x)}{W(y|x)} \\
&\leq \sum_{x\in\mathcal{X}_1(x^l), y\in \mathcal{Y}: V(y|x)\neq 0, N_{x}(x^l)\neq 0}\frac{N_{x}(x^l)}{l}\log\frac{V(y|x)}{W(y|x)} \\
&\overset{(b)}\leq \gamma_2  \sum_{x\in\mathcal{X}_1(x^l), y\in \mathcal{Y}: V(y|x)\neq 0, N_{x}(x^l)\neq 0}\left|\log\frac{V(y|x)}{W(y|x)}\right|\\
&\overset{(c)}\leq \gamma_2  \sum_{x\in\mathcal{X}_1(x^l), y\in \mathcal{Y}: V(y|x)\neq 0, N_{x}(x^l)\neq 0}\left|\log\frac{V(y|x)}{W(y|x)}\right|\cr
&\qquad\qquad+\sum_{x\in\mathcal{X}_1(x^l), y\in \mathcal{Y}: V(y|x)\neq 0, N_{x}(x^l)\neq 0}\frac{N_x(x^l)}{l}\left(V(y|x)\log\frac{V(y|x)}{W(y|x)} +\left|\log\frac{V(y|x)}{W(y|x)}\right|\right)\\
&\leq \sum_{x\in\mathcal{X}_1(x^l), y\in \mathcal{Y}}\frac{N_x(x^l)}{l}V(y|x)\log\frac{V(y|x)}{W(y|x)}+2\gamma_2  \sum_{x\in\mathcal{X}_1(x^l), y\in \mathcal{Y}: V(y|x)\neq 0, N_{x}(x^l)\neq 0}\left|\log\frac{V(y|x)}{W(y|x)}\right|, \label{eqn:x1}
\end{align}
where $(a)$ is because $V^l(y^l|x^l)\neq 0$ implies that $N_{x,y}(x^l,y^l)=0$ for $x$ and $y$ such that $V(y|x)=0$, $(b)$ is due to the definition of $\mathcal{X}_1(x^l)$,  and $(c)$ is because  the term in the brackets in the second summation is always non-negative. Next, we have
\begin{align}
&\sum_{x\in\mathcal{X}_2(x^l), y\in \mathcal{Y}}\frac{N_{x,y}(x^l, y^l)}{l}\log\frac{V(y|x)}{W(y|x)} \cr
&=\sum_{x\in\mathcal{X}_2(x^l), y\in \mathcal{Y}: V(y|x)\neq 0,N_{x}(x^l)\neq 0}\frac{N_x(x^l)}{l}\frac{N_{x,y}(x^l, y^l)}{N_x(x^l)}\log\frac{V(y|x)}{W(y|x)}\\
&\overset{(a)}\leq \sum_{x\in\mathcal{X}_2(x^l), y\in \mathcal{Y}}\frac{N_x(x^l)}{l}V(y|x)\log\frac{V(y|x)}{W(y|x)}\cr
&\qquad\qquad\qquad\qquad+\gamma_1\sum_{x\in\mathcal{X}_2(x^l), y\in \mathcal{Y}: V(y|x)\neq 0,N_{x}(x^l)\neq 0}\frac{N_x(x^l)}{l}\left|\log\frac{V(y|x)}{W(y|x)}\right|\\
&\leq \sum_{x\in\mathcal{X}_2(x^l), y\in \mathcal{Y}}\frac{N_x(x^l)}{l}V(y|x)\log\frac{V(y|x)}{W(y|x)}+\gamma_1\sum_{x\in\mathcal{X}_2(x^l),y\in \mathcal{Y}: V(y|x)\neq 0,N_{x}(x^l)\neq 0}\left|\log\frac{V(y|x)}{W(y|x)}\right|, \label{eqn:x2}
\end{align}
where $(a)$ is from the definitions of $\mathcal{X}_2(x^l)$ and $\calJ^{\gamma_{1},\gamma_{2}}_{x^l}$. 

By combining \eqref{eqn:x1} and \eqref{eqn:x2}, we obtain 
\begin{align}
&\sum_{x\in\mathcal{X},y\in \mathcal{Y}}\frac{N_{x,y}(x^l, y^l)}{l}\log\frac{V(y|x)}{W(y|x)}\cr
&\leq \sum_{x\in\mathcal{X},y\in \mathcal{Y}}\frac{N_x(x^l)}{l}V(y|x)\log\frac{V(y|x)}{W(y|x)}+(\gamma_1+2\gamma_2)\sum_{x\in\mathcal{X},y\in \mathcal{Y}: V(y|x)\neq 0,N_{x}(x^l)\neq 0}\left|\log\frac{V(y|x)}{W(y|x)}\right|\\
&\leq D(V\|W|P_{x^l})+(\gamma_1+2\gamma_2)\gamma'\label{eqn:ba}.
\end{align}
By substituting \eqref{eqn:ba} to \eqref{eqn:ratio}, the proof is completed. \end{proof}

\begin{proof}[Proof of Lemma~\ref{lem:cont}]
Because $W$ is assumed to be an output symmetric DMC, we have
\begin{align}
\max_{x\in \mathcal{X}} D(V_n^*(\cdot|x)\|W(\cdot|x))&=\max_{P\in\calP (\calX)} D(V_n^*\|W|P)\\
&=E^+(C-n^{-t}-2\delta_n)\\
&=E_{\mathrm{SP}}(C-n^{-t}-2\delta_n)\\
&\overset{(a)}= O(n^{-2t}), \label{eqn:max_bd}
\end{align}
where $(a)$ is from Lemma~\ref{lemma:sp}. Hence, for every $x\in \mathcal{X}$, 
\begin{align}
D(V_n^*(\cdot|x)\|W(\cdot|x))=\sum_y V_n^*(y|x)\log \frac{V_n^*(y|x)}{W(y|x)}=O(n^{-2t}). \label{eqn:dist}
\end{align}
Note that \eqref{eqn:dist} implies that if $W(y|x)=0$ for some $x\in \mathcal{X}$ and $y\in \mathcal{Y}$, then $V_n^*(y|x)=0$.\footnote{Note that if $W(y|x)= 0$ but $V_n^*(y|x)\neq 0$, $ V_n^*(y|x)\log \frac{V_n^*(y|x)}{W(y|x)}=\infty$ and thus \eqref{eqn:dist} does not hold. }

Now, let $\Delta_n(y|x):=|V_n^*(y|x)-W(y|x)|$ for all $x\in \mathcal{X}$ and $y\in \mathcal{Y}$. Then, for all $x\in \mathcal{X}$ and $y\in \mathcal{Y}$, $\Delta_n(y|x)=O(n^{-t})$. This can be proved using the Pinsker's inequality \cite[p.44]{CsiszarKorner:11}, i.e., for each $x\in \mathcal{X}$,
\begin{align}
\sum_{y\in \mathcal{Y}}\Delta_n(y|x)\leq \sqrt{2\ln 2 \cdot D(V_n^*(\cdot|x)\|W(\cdot|x))}.
\end{align}  
Then, we obtain 
\begin{align}
\sum_{(x,y):V_n^*(y|x)\ne 0 } \left| \log\frac{V_n^*(y|x)}{W(y|x)} \right|&\overset{(a)}=
\sum_{(x,y):V_n^*(y|x)\ne 0, W(y|x)\ne 0 } \left| \log\frac{V_n^*(y|x)}{W(y|x)} \right|\\
&\overset{(b)}\leq \sum_{(x,y):W(y|x)\ne 0 } \left( \frac{\Delta_n(y|x)}{W(y|x)} +O\Big(\frac{\Delta_n^2(y|x)}{W^2(y|x)}\Big)\right)\\
&=O(n^{-t}),
\end{align}
where $(a)$ is because for all $(x,y)$ such that $V_n^*(y|x) > 0$, it is true that $W(y|x) > 0$ due to  \eqref{eqn:dist} and $(b)$ is from Taylor's theorem, which completes the proof. 
\end{proof}


\end{document}